\documentclass[12pt,reqno]{amsart}
\usepackage{amsmath}
\usepackage{latexsym}
\usepackage{amsfonts}
\usepackage{amssymb}
\usepackage{color}
\usepackage{bbm}
\usepackage{graphicx}

\newtheorem{proposition}{Proposition}

\newtheorem{corollary}{Corollary}

\theoremstyle{definition}

\newtheorem{example}{Example}

\newcommand{\real}{\mathbb R} 
\newcommand{\complex}{\mathbb C} 
\newcommand{\nat}{\mathbb N} 
\newcommand{\half}{\tfrac{1}{2}} 
\newcommand{\mo}[1]{\left| #1 \right|} 
\newcommand{\abs}{\mo} 

\newcommand{\hi}{\mathcal{H}} 
\newcommand{\lh}{\mathcal{L(H)}} 
\newcommand{\trh}{\mathcal{T(H)}} 
\newcommand{\sh}{\mathcal{S(H)}} 
\newcommand{\eh}{\mathcal{E(H)}} 
\newcommand{\ip}[2]{\left\langle\,#1\,|\,#2\,\right\rangle} 
\newcommand{\ket}[1]{|#1\rangle} 
\newcommand{\kb}[2]{|#1\rangle\langle#2|} 
\newcommand{\no}[1]{\left\|#1\right\|} 
\newcommand{\tr}[1]{\mathrm{tr}\left[#1\right]} 
\newcommand{\ran}{\mathrm{ran}} 
\newcommand{\id}{\mathbbm{1}} 

\renewcommand{\span}{\mathrm{span}} 

\newcommand{\mdc}{\mathcal{M}_{d}(\complex)} 
\newcommand{\mkc}{\mathcal{M}_{d_k}(\complex)} 

\newcommand{\Ao}{\mathsf{A}}
\newcommand{\Bo}{\mathsf{B}}
\newcommand{\Go}{\mathsf{G}}
\newcommand{\Po}{\mathsf{P}}

\newcommand{\Ii}{\mathcal{I}} 
\newcommand{\Ic}{\Ii_{\Omega}} 
\newcommand{\IiL}{\Ii^L} 
\newcommand{\IcL}{\Ic^L} 

\newcommand{\fix}[1]{F(#1)} 

\newcommand{\1}{\mathbbm{1}}

\newcommand{\be}{\begin{equation}}
\newcommand{\ee}{\end{equation}}
\newcommand{\bea}{\begin{eqnarray}}
\newcommand{\eea}{\end{eqnarray}}

\begin{document}

\title[Non-disturbing quantum measurements]{Non-disturbing quantum measurements}

\author[Heinosaari]{Teiko Heinosaari}
\address{Teiko Heinosaari, Niels Bohr Institute, Blegdamsvej 17, 2100 Copenhagen, Denmark}
\email{heinosaari@nbi.dk}

\author[Wolf]{Michael M. Wolf}
\address{Michael M. Wolf, Niels Bohr Institute, Blegdamsvej 17, 2100 Copenhagen, Denmark}
\email{wolf@nbi.dk}

\begin{abstract} We consider pairs of discrete quantum observables (POVMs) and analyze the relation between the notions of non-disturbance, joint measurability and commutativity. We specify conditions under which these properties coincide or differ---depending for instance on the interplay between the number of outcomes and the Hilbert space dimension or on algebraic properties of the effect operators. We also show that (non-)disturbance is in general not a symmetric relation and that it can be decided and quantified by means of a semidefinite program. 
\end{abstract}

\maketitle

\section{Introduction}\label{sec:intro}

One of the main features of quantum mechanics is that measurements of different observables usually disturb each other. 
This property often comes along with non-commutativity or the impossibility of measuring observables simultaneously. 
Strictly speaking, however, \emph{non-disturbance, joint measurability} and \emph{commutativity} are different concepts and it is the aim of the present paper to clarify their precise relation. 

In general all these notions turn out to be different, but we will see that some coincide with others under certain conditions involving  for instance the interplay between the number of measurement outcomes and the Hilbert space dimension or algebraic properties of the effect operators. 

In this work we concentrate on discrete observables.
Our investigation is organized as follows.
A brief overview of the relevant concepts is given in Section \ref{sec:preli}.
The disturbance of one measurement w.r.t. another is then studied in Section \ref{sec:pairs}, where it is shown that (i) non-disturbance is not a symmetric relation, (ii) it is equivalent to commutativity if the second measurement has sufficiently many independent outcomes but (iii) inequivalent to joint measurability and commutativity in general. 
In Section \ref{sec:fkm} we discuss measurements which do not disturb themselves, i.e., measurements of the first kind, and their relation to  repeatability and commutativity. Finally in Section \ref{sec:sdp} we argue that (non-)disturbance can be decided and quantified efficiently by means of a semidefinite program.

\section{Preliminaries}\label{sec:preli}

This section will fix some notations and introduce the basic concepts.
Let $\hi$ be a complex Hilbert space, either finite or countably infinite dimensional.
We denote by $\lh$ the set of bounded linear operators and by $\trh\subseteq\lh$ the set of trace class operators on $\hi$.
A positive operator $\varrho\in\trh$ having trace one is a \emph{density operator}, also referred to as \emph{state}, and we denote by $\sh$ the set of all states.

\subsection{Observables}

Observables are generally described by \emph{positive operator valued measures} (POVMs).
In this work we only consider discrete observables.
Therefore, an observable is characterized by a finite or countably infinite collection of positive operators $\Ao\equiv\{\Ao_x\}$ satisfying $\sum_x \Ao_x = \id$. Here the sum runs over all $x\in\Omega_\Ao$, where the set $\Omega_\Ao$ is the collection of all possible measurement outcomes.
Whenever convenient we take $\Omega_\Ao \subseteq \nat$ as the labeling of the outcomes is irrelevant in our investigation.
If a system is prepared in a state $\varrho$, then a measurement of an observable $\Ao$ will lead to an outcome $x$ with probability $\tr{\varrho\Ao_x}$.

A selfadjoint operator $E\in\lh$ satisfying $0\leq E \leq \id$ is called an \emph{effect} and we denote the set of all effects by $\eh$.
Note that $\Ao_x\in \eh$, i.e., the elements of a POVM are effects.

\subsection{Joint measurability and commutativity}

Given two observables $\Ao$ and $\Bo$, we say that they are \emph{jointly measurable} if there exists a third observable $\Go$ with $\Omega_\Go=\Omega_\Ao \times\Omega_\Bo$ and satisfying $\sum_x\Go_{x,y}=\Bo_y$ for all $y$ and $\sum_y \Go_{x,y}=\Ao_x$ for all $x$.
In other words, $\Ao$ and $\Bo$ correspond to the `marginals' of $\Go$.

A particular case of jointly measurable pairs of observables  $\Ao$ and $\Bo$ are those which commute, i.e.,
\begin{equation*}
[\Ao_x,\Bo_y]=0
\end{equation*}
for all $x\in\Omega_{\Ao}$, $y\in\Omega_{\Bo}$.
In this case we can set $\Go_{x,y}=\Ao_x\Bo_y$ which defines an observable $\Go$ on the product set $\Omega_{\Ao}\times\Omega_{\Bo}$ since the commutativity of $\Ao$ and $\Bo$ guarantees that the operators $\Go_{x,y}$ are positive.
Note that when talking about commuting pairs of observables we do not necessarily require that the effects within each observable are commuting, i.e., $[\Ao_x,\Ao_{x'}]$ may be non-zero for $x\neq x'$.

Two observables $\Ao$ and $\Bo$ can be jointly measurable even if they do not commute.
The relation of being jointly measurable is also qualitatively different from commutativity.
For instance, even if all partitionings of $\Ao$ and $\Bo$ into two outcome observables are jointly measurable, it may happen that $\Ao$ and $\Bo$ are not \cite{HeReSt08}.
We also recall that not being jointly measurable is closely related and sometimes provably equivalent to the possibility of detecting 'quantum non-locality', i.e., the ability of violating a Bell inequality \cite{WoPeFe09}.

\subsection{Instruments}

An observable describes the statistics of the outcomes of a measurement but leaves open how the measurement alters the quantum state.
In order to discuss this we need the concept of an \emph{instrument} \cite{QTOS76}.
An instrument which implements an observable $\Ao$ is a collection of completely positive linear maps $\Ii\equiv\{\Ii_x\}$ on $\trh$ which satisfy $\Ii_x^\ast(\1)=\Ao_x$ for every $x$.
Here the \emph{adjoint map} $\Ii_x^\ast$ is defined via the usual trace duality $\tr{X \Ii_x(Y)}=\tr{\Ii_x^\ast(X)Y}$ for all $X\in\lh,Y\in\trh$.
In other words, $\Ii_x^\ast$ and $\Ii_x$ correspond to the Heisenberg and Schr\"odinger pictures, respectively.

If an observable $\Ao$ is implemented via an instrument $\Ii$, then $\Ii_x(\varrho)$ is the unnormalized state after having obtained the measurement outcome $x$ upon an input state $\varrho$.
Note that $\tr{\Ii_x(\varrho)}=\tr{\varrho\Ao_x}$ is the probability for this to happen.
 We always assume that the output system has the same dimension as the input system.
 
If we ignore the measurement outcome, an instrument $\Ii$ transforms an input state $\varrho$ to an (unconditional) output state
\begin{equation*}
\Ic(\varrho):=\sum_x\Ii_x(\varrho) \, .
\end{equation*}
The map $\Ic$ is a \emph{quantum channel}, i.e., a completely positive, trace-preserving linear map on $\trh$.
The dual map $\Ic^\ast$ is completely positive, identity preserving linear map on $\lh$.
It is also $\sigma$-weakly continuous.

Evidently, many different instruments correspond to the same observable.
A particular implementation of an observable $\Ao$ is given by its \emph{L\"uders instrument} $\IiL$, defined as
\begin{equation*}
\IiL_x(\varrho)=\Ao_x^{1/2} \varrho \Ao_x^{1/2} \, .
\end{equation*}
It is also easy to give examples of other instruments.
For instance, fix a state $\xi_x$ for each outcome $x\in\Omega_\Ao$.
Then the formula
\begin{equation}
\Ii_x(\varrho) = \tr{\varrho \Ao(x)}\xi_x
\end{equation}
defines an instrument implementing $\Ao$.

\subsection{Non-disturbing measurements}

Given two observables $\Ao$ and $\Bo$ we say that $\Ao$ \emph{can be measured without disturbing} $\Bo$ if there exists an instrument $\Ii$ which implements $\Ao$ and for which
\begin{equation}\label{eq:A-not-disturb-B-Sch}
\tr{\Ic(\varrho)\Bo_y}=\tr{\varrho\Bo_y} \quad \forall \varrho\in\sh,y\in\Omega_\Bo \, .
\end{equation}
This means that the measurement statistics of $\Bo$ are the same for all pairs of an input state $\varrho$ and an output state $\Ic(\varrho)$.

We can write the  non-disturbance condition \eqref{eq:A-not-disturb-B-Sch} in an equivalent form
\begin{equation}\label{eq:A-not-disturb-B}
\Ic^*(\Bo_y)=\Bo_y\quad\forall y\in\Omega_\Bo \, .
\end{equation}
Hence, this is to say all the effects $\Bo_y$ are \emph{fixed points} of $\Ic^\ast$.
We denote by $\fix{\Ic^\ast}\subseteq\lh$ the set of fixed points of $\Ic^\ast$. This is clearly a linear subspace of $\lh$.
Moreover, since $\Ic^\ast$ is $\sigma$-weakly continuous, $\fix{\Ic^\ast}$ is $\sigma$-weakly closed.

If $\Ao$ and $\Bo$ commute, then a non-disturbing measurement can be achieved by the L\"uders instrument implementing $\Ao$ since
\begin{eqnarray*}
\tr{\IcL(\varrho)\Bo_y} &=& \sum_{x} \tr{\Ao_x^{1/2} \varrho \Ao_x^{1/2} \Bo_y} =  \sum_{x} \tr{\Ao_x\varrho\Bo_y} \\
&=& \tr{\varrho \Bo_y} \, .
\end{eqnarray*}
We recall that if the Hilbert space $\hi$ is finite dimensional, then the L\"uders instrument implementing $\Ao$ does not disturb $\Bo$ if and \emph{only if} they commute \cite{BuSi98}.
In an infinite dimensional Hilbert space this is not generally true; there exist non-commuting observables $\Ao$ and $\Bo$ such that the L\"uders measurement of $\Ao$ does not disturb $\Bo$  \cite{ArGhGu02}, \cite{WeJu09}.

To give a class of examples of pairs in which a non-disturbing measurement is not possible, suppose that $\Bo$ is an informationally complete observable.
This means that the probabilities $\tr{\varrho\Bo_y}$ uniquely determine every state $\varrho$.
It is then clear from the non-disturbance condition \eqref{eq:A-not-disturb-B-Sch} that $\Ic(\varrho)=\varrho$ for every state $\varrho$.
However, every non-trivial observable $\Ao$ necessarily perturbs at least some state and therefore also disturbs any informationally complete observable $\Bo$.

If $\Ao$ can be implemented via an instrument $\Ii$ which does not disturb $\Bo$, then $\Ao$ and $\Bo$ are jointly measurable.
This is quite obvious since their sequential measurement gives the measurement statistics of both $\Ao$ and $\Bo$ without any perturbation.
Formally,  we can set $\Go_{x,y}=\Ii_x^\ast(\Bo_y)$ which defines a joint observable $\Go$ for $\Ao$ and $\Bo$.

In Subsection \ref{sec:two-outcome} we will see an explicit example of two observables $\Ao$ and $\Bo$ which are jointly measurable despite the fact that one necessarily disturbs the other.
In that case, there is no instrument $\Ii$ implementing $\Ao$ and satisfying $\Ic^\ast(\Bo_y)=\Bo_y$ for all $y$.
Hence, the joint measurement cannot be implemented sequentially by first measuring $\Ao$ and then $\Bo$.

Let us remark that if two observables are jointly measurable, there is always a sequential implementation which encodes the outcomes of the $\Bo$-measurement into a quantum system and then recovers them again by a final measurement. 
The latter is then, however, generally different from $\Bo$ and one may have to increase the dimension of the Hilbert space for the encoding.
An instrument for this kind of sequential implementation can be chosen to be
\begin{equation*}
\Ii_x(\varrho) = \sum_y \tr{\varrho \Go_{x,y}} \kb{y}{y} \, ,
\end{equation*}
where $\Go$ is a joint observable of $\Ao$ and $\Bo$ and $\{\ket{y}\}$ is an orthonormal basis.
The instrument $\Ii$ implements $\Ao$, while a subsequent measurement in the basis $\ket{y}$ yields $\Bo$.

\subsection{Sharp observables}\label{sec:sharp}

An observable $\Ao$ is called \emph{sharp} if each effect $\Ao_x$ is a projection.
It is well known that if we have two observables $\Ao,\Bo$ and (at least) one of them is sharp, then the three relations - commutativity, joint measurability and non-disturbance - are equivalent.
In the following we note two slightly stronger results.

\begin{proposition}\label{prop:joint->commute}
Let $\Ao$ and $\Bo$ be two jointly measurable observables.
Suppose that $\Bo_y$ is a projection.
Then $[\Ao_x,\Bo_y]=0$ for all $x\in\Omega_\Ao$.
\end{proposition}

\begin{proof}
Let $\Go$ be a joint observable for $\Ao$ and $\Bo$.
Since $\sum_{x} \Go_{xy} = \Bo_y$, we conclude that $\Go_{xy}\leq\Bo_y$ for every $x$.
But as $\Bo_y$ is a projection, this implies that $\Go_{xy}\Bo_y = \Bo_y\Go_{xy}=\Go_{xy}$ for every $x$.
Also $\id - \Bo_y$ is a projection and $\sum_{y'\neq y}\sum_{x} \Go_{xy'} = \id - \Bo_y$.
We thus get $\Go_{xy'}(\id-\Bo_y) = (\id-\Bo_y)\Go_{xy'}=\Go_{xy'}$ for every $x$ and every $y'\neq y$, implying that $\Go_{xy'}\Bo_y = \Bo_y\Go_{xy'}=0$.
We conclude that $\Bo_y$ commutes with all effects $\Go_{xy'}$, therefore also with $\Ao_x = \sum_{y'} \Go_{xy'}$.
\end{proof}

\begin{proposition}\label{prop:nondist->commute}
Let $\Ao$ and $\Bo$ be two observables.
Suppose that for some $y$, the effect $\Bo_y$ is proportional to a projection.
If an instrument $\Ii$ implements $\Ao$ and $\Ic^\ast(\Bo_y) = \Bo_y$,
then $[\Ao_x,\Bo_y]=0$ for all $x\in\Omega_\Ao$.
\end{proposition}

\begin{proof}
The proof follows Lemma 3.3 in \cite{BrJoKiWe00}.
Denote by $P$ the projection onto the support of $\Bo_y$ and $P^\perp=\1-P$.
The assumption means that $\Bo_y = cP$ for some $c > 0$, and $\Ic^\ast(\Bo_y) = \Bo_y$ thus implies that $\Ic^\ast(P) =P$.
Moreover, since $\Ic^\ast(\id)=\id$ we also have $\Ic^\ast(P^\perp)=P^\perp$.

Since $\Ii_x$ is completely positive we can make use
of its Kraus decomposition $\Ii_x(\cdot)=\sum_\alpha K_{\alpha,x}\cdot
K_{\alpha,x}^\ast$.
If we insert this into the equation $\Ic^\ast(P)=P$  and multiply with $P^\perp$ from the left and the
right, we get
\be
\sum_{\alpha,x}P^\perp K_{\alpha,x}^\ast P^2 K_{\alpha,x}P^\perp=0
\label{eq:profpropprop}.\ee
Since all the summands are positive this implies $P K_{\alpha,x}
P^\perp=0$ and therefore $P K_{\alpha,x} P = P  K_{\alpha,x}$.
Interchanging $P$ and $P^\perp$ in the argument gives $P K_{\alpha,x}
P =   K_{\alpha,x} P$ and thus $[P, K_{\alpha,x}]=0$.
Taking the adjoint of this equation shows that $P$ also commutes with $K_{\alpha,x}^\ast$,  so
that $[\Ao_x,P]=0$ and finally $[\Ao_x,\Bo_y]=0$.
\end{proof}

\section{Pairs of quantum observables}\label{sec:pairs}

In this section we investigate the non-disturbance relation for pairs of observables.
We start by demonstrating that the non-disturbance criterion differs from the joint measurability and commutativity criteria even for pairs of two-outcome observables (Subsec. \ref{sec:two-outcome}).
However, with some additional requirements non-disturbance reduces to commutativity.
This happens, for instance, in the case of qubit observables (Subsec. \ref{sec:qubit}).
The overall picture is schematically summarized in Fig. \ref{fig:1}.

\begin{figure}
\begin{center}
\includegraphics[width=10cm]{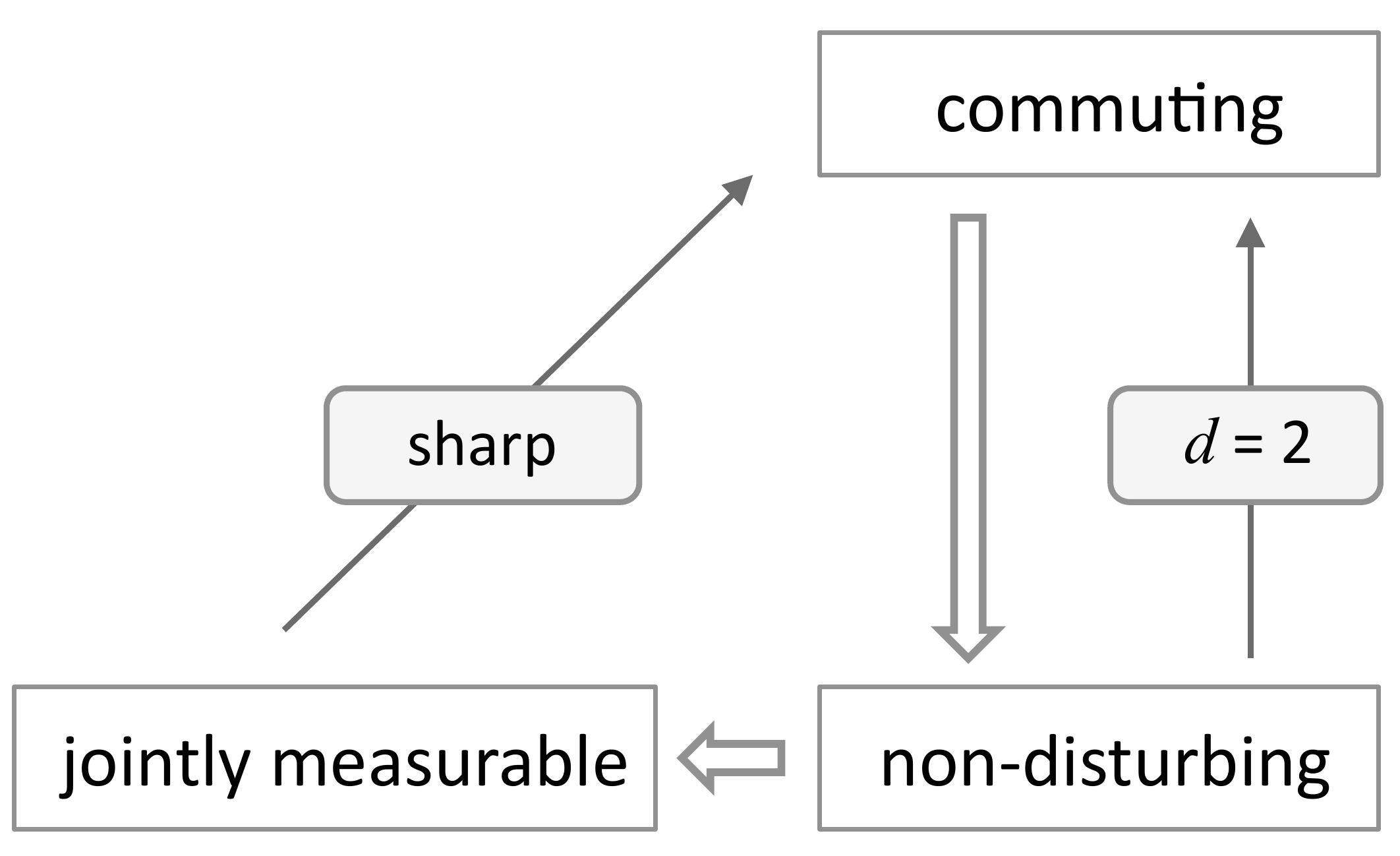}
\end{center}
\caption{Implications between the different concepts. The gray boxes indicate the additional conditions under which the implications holds. \label{fig:1}}
\end{figure}

We show that even if an observable $\Ao$ can be measured without disturbing another observable $\Bo$, the converse need not hold (Subsec. \ref{sec:not-symmetric}). This means that non-disturbance is not a symmetric relation, unlike commutativity and joint measurability. Finally, we discuss the disturbance caused by a rank-1 observable (Subsec. \ref{sec:rank1}).

\subsection{Two-outcome measurements}\label{sec:two-outcome}

In the simplest case an observable $\Ao$ has only two outcomes.
It is then determined by a single effect $\Ao_1$, since the normalization requires that $\Ao_2=\id-\Ao_1$.
Clearly, a two-outcome observable is commutative.

Let $\Ao$ and $\Bo$ be sharp two-outcome observables.
We assume that $\Ao$ and $\Bo$ do not commute, hence they are not jointly measurable since these concepts coincide for sharp observables.

For all $\mu,\nu\in(\half,1]$, we define \emph{coarse-grainings} of $\Ao$ and $\Bo$ by
\begin{equation*}
\Ao^{\mu}_1 =  \mu\Ao_1 +(1-\mu)\Ao_2 \, , \quad  \Ao^{\mu}_2 =  (1-\mu) \Ao_1 +\mu \Ao_2
\end{equation*}
and
\begin{equation*}
\Bo^{\nu}_1 =  \nu\Bo_1 +(1-\nu)\Bo_2 \, , \quad  \Bo^{\nu}_2 =  (1-\nu) \Bo_1 +\nu \Bo_2\, .
\end{equation*}
In this way we get observables $\Ao^{\mu}$ and $\Bo^{\nu}$, which we regard as approximate versions of $\Ao$ and $\Bo$, respectively.
The numbers $\mu$ and $\nu$ quantify the levels of approximation.

If $\mu$ and $\nu$ are small enough, then $\Ao^{\mu}$ and $\Bo^{\nu}$ are jointly measurable.
For instance, suppose that $\mu=\nu$ and define four operators by
\begin{eqnarray*}
&& \Go_{1,1} = (1-\mu)\,\id \\
&& \Go_{1,2} = (2\mu-1)\Ao_1 \\
&& \Go_{2,1} = (2\mu-1)\Bo_1 \\
&& \Go_{2,2} = \mu\, \id - (2\mu-1) (\Ao_1 + \Bo_1) \, .
\end{eqnarray*}
The first three operators are clearly positive, and taking into account that $0\leq \Ao_1 + \Bo_1 \leq 2 \id$ we see that also the fourth operator is positive if $\half < \mu \leq \frac{2}{3}$.
Under this condition $\Go$ is a joint observable for $\Ao^{\mu}$ and $\Bo^{\mu}$.
In conclusion, if $\half < \mu \leq \frac{2}{3}$ then $\Ao^{\mu}$ and $\Bo^{\mu}$ are jointly measurable.

We claim that $\Ao^{\mu}$ and $\Bo^{\nu}$ cannot be measured without disturbing each other, no matter how the numbers $\mu$ and $\nu$ are chosen from the interval $(\half,1]$.
To prove this, let us first notice that
\begin{equation*}
\Ao^{\mu}_1 = (2\mu-1) \Ao_1 + (1-\mu) \id \, , \quad \Bo^{\nu}_1 = (2\nu-1) \Bo_1 + (1-\nu) \id \, .
\end{equation*}
Hence,
\begin{equation*}
[\Ao^{\mu}_1,\Bo^{\nu}_1] = (2\mu-1)(2\nu-1) [\Ao_1,\Bo_1] \, ,
\end{equation*}
showing that $\Ao^{\mu}$ and $\Bo^{\nu}$ do not commute for any $\mu,\nu\in (\half,1]$.
In particular, $\Ao^{\mu}$ and the sharp observable $\Bo$ are not jointly measurable according to Prop.\ref{prop:nondist->commute}.

Let us make a counter assumption that there exists an instrument $\Ii$ implementing $\Ao^{\mu}$ and not disturbing $\Bo^{\nu}$, i.e.,
\begin{equation}
\Ic^\ast(\Bo^{\nu}_y) = \Bo^{\nu}_y \qquad \ y=1,2 \, .
\end{equation}
But since $\Ic^\ast(\id) = \id$
and $\Bo_y$ is a linear combination of $\id$ and $\Bo^{\nu}_y$, it follows that
\begin{equation}
\Ic^\ast(\Bo_y) = \Bo_y \qquad \ y=1,2 \, .
\end{equation}
However, this non-disturbance condition cannot hold as $\Ao^{\mu}$ and $\Bo$ are not jointly measurable. So joint measurability of the above observables does never imply non-disturbance.

Let us proceed by giving an example of non-commuting observables $\Ao$ and $\Bo$ such that there exists an $\Ao$-measurement not disturbing $\Bo$.
This example originates from Remark 2 in \cite{Arveson72}.

Let $\hi=\complex^3$.
We define observables $\Ao$ and $\Bo$ as
\begin{equation*}
\Ao_1= \frac{1}{4}
\begin{pmatrix}
 2 & 0 &  -\sqrt{2}\\
0 & 4 & 0\\
-\sqrt{2} & 0 &  3
\end{pmatrix}
\, , \quad
\Ao_2= \frac{1}{4}
\begin{pmatrix}
 2 & 0 &  \sqrt{2}\\
0 & 0 & 0\\
 \sqrt{2} & 0 & 1
\end{pmatrix}
\end{equation*}
and
\begin{equation}\label{eq:B-diagonal}
\Bo_1=\frac{1}{2}
\begin{pmatrix}
 2 & 0 &  0 \\
0 & 0 & 0\\
0 & 0 &  1
\end{pmatrix}
\, , \quad
\Bo_2=\frac{1}{2}
\begin{pmatrix}
0 & 0 & 0\\
0 & 2 & 0\\
0 & 0 & 1
\end{pmatrix}
\, .
\end{equation}
Then $\Ao$ and $\Bo$ do not commute.

However, there is an instrument implementing $\Ao$ which does not disturb $\Bo$.
We set
\begin{eqnarray*}
&K_1& =\frac{1}{2}
\begin{pmatrix}
 \sqrt{2} & 0 & 0 \\
0 & 0 & 0\\
-1 & 0 & 0
\end{pmatrix}
\, , \quad
K_2 = \frac{1}{10}
\begin{pmatrix}
0 & 0 & 0\\
0 & -\sqrt{10} & 0\\
0 & 2\sqrt{10} & 0
\end{pmatrix}
\, , \\ 
&K_3& = \frac{1}{2}
\begin{pmatrix}
0 & 0 & 0\\
0 & \sqrt{2} & 0\\
0 & 0 & 0
\end{pmatrix}
\, , \quad
K_4 =\frac{1}{20}
\begin{pmatrix}
0 & 0 & 0\\
0 & 4\sqrt{10} & 0\\
0 & 2\sqrt{10} & 0
\end{pmatrix}
\, , \\
& K_5 &=\frac{1}{2}
\begin{pmatrix}
 \sqrt{2} & 0 & 0 \\
0 & 0 & 0\\
 1 & 0 & 0
\end{pmatrix} 
\end{eqnarray*}
and define $\Ii_1^\ast (\cdot) = \sum_{j=1}^4 K_j \cdot K_j^\ast$ and $\Ii_2^\ast (\cdot) = K_5 \cdot K_5^\ast$.
It is straighforward to check that $\Ii$ satisfies the non-disturbance condition \eqref{eq:A-not-disturb-B}.

\subsection{When does non-disturbance reduce to commutativity?}\label{sec:qubit}

It is a fundamental fact of quantum theory that every measurement perturbs the system.
Therefore, we expect that in a sequence of non-disturbing measurements, the second measurement cannot be too informative since otherwise we would detect the perturbation caused by the first measurement.
In the following we give some precise conditions for this intuitive idea.

For each observable $\Bo$, we denote by $\span\Bo$ the linear subspace in $\lh$ generated by the set $\{\Bo_y : y\in\Omega_\Bo \}$, i.e.,
\begin{equation*}
\span\Bo = \{ \sum_y c_y \Bo_y \mid c_y \in \complex, c_y\neq 0\  \textrm{for finitely many $y$} \} \, .
\end{equation*}
By $\overline{\span}\Bo$ we denote the closure of $\span\Bo$ in the $\sigma$-weak operator topology.
(Clearly, if $\dim\hi < \infty$, then $\overline{\span}\Bo=\span\Bo$.)

\begin{proposition}\label{prop:square}
Suppose that an observable $\Bo$ has the following property:
\begin{equation}\label{eq:square}
\forall y \in \Omega_\Bo:\quad \Bo_y^2 \in \overline{\span}\Bo \, .
\end{equation}
Then it is possible to measure an observable $\Ao$ without disturbing $\Bo$ if and only if $\Ao$ and $\Bo$ commute.
\end{proposition}

\begin{proof}
Let $\Ii$ be an instrument which implements $\Ao$ and does not disturb $\Bo$.
The set $\fix{\Ic^\ast}$  of the fixed points of $\Ic^\ast$ is a $\sigma$-weakly closed linear subspace of $\lh$.
Hence, from \eqref{eq:square} follows that every $\Bo_y^2$ is a fixed point of $\Ic^\ast$.

Let $\Ic(\varrho) = \sum_{j} K_{j} \varrho K_{j}^\ast$ be a Kraus decomposition for $\Ic$.
Then for every $y \in \Omega_\Bo$, we get
\begin{equation*}
\sum_{j} [K_j,\Bo_y]^\ast\ [K_j,\Bo_y] = \Ic^\ast(\Bo_y^2)+ \Bo_y^2-\Bo_y\Ic^\ast(\Bo_y)- \Ic^\ast(\Bo_y)\Bo_y = 0
\end{equation*}
and therefore $[K_j,\Bo_y]=0$ for each $j$.
This implies that $\Ao$ and $\Bo$ commute.
\end{proof}

To give a class of examples where the condition \eqref{eq:square} holds, suppose that $\Bo$ is a classical coarse-graining of a sharp observable $\Po$ in the sense that there is a stochastic square matrix $M$ such that $\Bo_y=\sum_{y'} M_{yy'} \Po_{y'}$.
Suppose further that $M$ is invertible (but the inverse need not be a stochastic matrix).
Then $\Po_{y'}=\sum_{y} M^{-1}_{y'y} \Bo_y$, implying that  $\Bo_y^2 \in \span\Bo$ for every $y\in\Omega_\Bo$.
The condition of $M$ being invertible means that $\Bo$ and $\Po$ are \emph{informationally equivalent} \cite{AlDo76}; $\Bo$ gives different measurement outcome distributions for two states $\varrho_1$ and $\varrho_2$ if and only if their measurement outcome distributions are different in a $\Po$-measurement.

The remaining results in this subsection rest on the following.

\begin{proposition}\label{prop:full-rank}
Let $\Ao$ be an observable and $\Ii$ an instrument implementing $\Ao$.
Suppose that the channel $\Ic$ has a full rank fixed point $\varrho$.
If $\Ii$ does not disturb an observable $\Bo$, then $\Ao$ and $\Bo$ commute.
\end{proposition}

\begin{proof}
As shown in Lemma 3.4 in \cite{BrJoKiWe00}, the set of fixed points $\fix{\Ic^\ast}$ forms an algebra if $\Ic$ has a full rank fixed point.
The claim then follows from Prop. \ref{prop:square}.
\end{proof}

In the rest of this subsection we assume that $\hi$ is a finite dimensional Hilbert space.
We can then identify $\trh$ and $\lh$ with the set $\mdc$ of $d\times d$ complex matrices, where $d=\dim\hi$.

Let $\Ii$ be an instrument.
By fixing a basis in $\mdc$, we can consider the channel $\Ic$ as a matrix acting on the $d^2$-dimensional vector space $\mdc$.
In this way, the fixed points of $\Ic$ are the right eigenvectors with eigenvalue $1$, while the fixed points of the dual channel $\Ic^\ast$ are the left eigenvectors with eigenvalue $1$.
In particular, the subspaces $\fix{\Ic}$ and $\fix{\Ic^\ast}$ consisting of the fixed points of $\Ic$ and $\Ic^\ast$ have the same dimension, $\dim\fix{\Ic} = \dim\fix{\Ic^\ast}$.

To formulate the following statement, we denote by $\dim\Bo$ the dimension of the linear subspace $\span\Bo\subset\mdc$.
Roughly speaking, $\dim\Bo$ is the number of independent measurement outcomes obtained in a $\Bo$-measurement.

\begin{proposition}\label{prop:dim}
Let $\Bo$ be an observable such that
\begin{equation}\label{eq:dim}
\dim\Bo \geq (d-1)^2+1 \, .
\end{equation}
It is possible to measure an observable $\Ao$ without disturbing $\Bo$ if and only if $\Ao$ and $\Bo$ commute.
\end{proposition}

\begin{proof}
Let $\Ii$ be an instrument which implements $\Ao$ and does not disturb $\Bo$.
We will show that the condition \eqref{eq:dim} implies that $\Ic$ has a full rank fixed point.
The claim then follows from Proposition \ref{prop:full-rank}.

As proved in \cite{Lindblad99} (see also \cite{Wolf10}), there is a unitary matrix $U$ and a set of states $\varrho_k$ such that
\begin{equation}
\fix{\Ic} = U \left( 0 \oplus \bigoplus_{k=1}^K \mkc \otimes \varrho_k \right) U^\ast
\end{equation}
for an appropriate decomposition of the Hilbert space $\complex^d = \complex^{d_0}\oplus\bigoplus_k\complex^{d_k}\otimes\complex^{m_k}$.
In particular, there are natural numbers $d_1,\ldots,d_K$ and $m_1,\ldots,m_K$ such that $\dim\fix{\Ic} = \sum_k d_k^2$ and $\sum_k d_k m_k \leq d$.
For convenience, we may assume that $d_k\geq d_{k+1}$ for all $k=1,\ldots,K-1$.

It follows from this decomposition of $\fix{\Ic}$ that $\dim\fix{\Ic}$ can take only some specific values.
Clearly, the largest value is $\dim\fix{\Ic} = d^2$ and then $\fix{\Ic}=\mdc$, thus $\Ic$ clearly has a full rank invariant state.
The second largest value is $\dim\fix{\Ic} = (d-1)^2 +1$ and this means that $d_1=d-1$ and $d_2=m_1=m_2=1$.
In this case $\id\in\fix{\Ic}$, hence $\Ic$ has a full rank fixed point.
Since $\dim \fix{\Ic} \geq \dim\Bo$, the claim follows.
\end{proof}

It is a direct consequence of Proposition \ref{prop:dim} that for qubit observables (i.e. $d=2$) non-disturbance and commutativity are equivalent conditions.
We find it useful to give also a simplified proof of this fact.

\begin{proposition}\label{prop:qubit}
Let $\dim\hi=2$.
For two observables $\Ao$ and $\Bo$, the following conditions are equivalent:
\begin{itemize}
	\item[(i)] It is possible to measure $\Ao$ without disturbing $\Bo$.
	\item[(ii)]   It is possible to measure $\Bo$ without disturbing $\Ao$.
	\item[(iii)] $\Ao$ and $\Bo$ commute.
\end{itemize}
\end{proposition}

\begin{proof}
Let $\Ii$ be an instrument which implements $\Ao$ without disturbing $\Bo$.
If the channel $\Ic$ has a full rank fixed point, then $\Ao$ and $\Bo$ commute by Proposition \ref{prop:full-rank}.

So let us then assume that $\Ic$ does not have a full rank fixed point.
This implies that $\dim \fix{\Ic} =1$, hence also $\dim \fix{\Ic^\ast} =1$.
But $\Ic^\ast(\id)=\id$, and therefore each $\Bo_y$ is a scalar multiple of the identity operator $\id$.
Hence, $\Bo$ commutes with $\Ao$.
\end{proof}

\subsection{Non-disturbance is not symmetric}\label{sec:not-symmetric}

In the following we demonstrate that the non-disturbance relation is not symmetric; there exist observables $\Ao$ and $\Bo$ such that every measurement of $\Bo$ disturbs $\Ao$ while a suitable $\Ao$-measurement does not disturb $\Bo$.
Again this example originates from Remark 2 in \cite{Arveson72}.
The construction requires that $\dim\hi\geq 3$, and we have seen in Section \ref{sec:qubit} that for $\dim\hi=2$ non-disturbance is equivalent to commutativity, hence a symmetric relation.

Let $\hi=\complex^3$.
We choose $\Bo$ and $K_1,\ldots,K_5$ as in the end of Subsec. \ref{sec:two-outcome}.
We take $\Ao$ to be the five outcome observable $\Ao_x=K_xK_x^\ast$, $x=1,\ldots,5$.
The effects are thus
\begin{eqnarray*}
\Ao_1 &=& \frac{1}{4}
\begin{pmatrix}
 2 & 0 &  -\sqrt{2}\\
0 & 0 & 0\\
-\sqrt{2} & 0 & 1
\end{pmatrix}
\, , \quad
\Ao_2 =\frac{1}{10}
\begin{pmatrix}
0 & 0 & 0\\
0 & 1 & -2\\
0 & -2 & 4
\end{pmatrix} 
\, ,
\\
\Ao_3 &=& \frac{1}{2}
\begin{pmatrix}
0 & 0 & 0\\
0 & 1 & 0\\
0 & 0 & 0
\end{pmatrix}
\, ,  \qquad\qquad\ 
\Ao_4 = \frac{1}{10}
\begin{pmatrix}
0 & 0 & 0\\
0 & 4 & 2\\
0 & 2 & 1
\end{pmatrix}
\, , 
\\
\Ao_5 &=& \frac{1}{4}
\begin{pmatrix}
 2 & 0 &  \sqrt{2}\\
0 & 0 & 0\\
 \sqrt{2} & 0 & 1
\end{pmatrix}
\, . 
\end{eqnarray*}
The instrument $\Ii_x^\ast (\cdot) = K_x \cdot K_x^\ast$ implements $\Ao$ and does not disturb $\Bo$.
However, the matrices $\Ao_1,\ldots,\Ao_5$ are linearly independent, implying that $\dim\Ao = 5$.
As $\Ao$ and $\Bo$ do not commute, it follows from Proposition \ref{prop:dim} that all $\Bo$-measurements disturb $\Ao$.

\subsection{Rank-1 observables}\label{sec:rank1}

An effect $E\in\eh$ is \emph{rank-1} if there is a one-dimensional projection $P$ and a number $0<e\leq 1$ such that $E=eP$.
A discrete observable $\Ao$ is called \emph{rank-1 observable} if each effect $\Ao_x$ is rank-1.
Rank-1 observables form an important subset of all observables.

From Proposition \ref{prop:nondist->commute} we conclude the following.

\begin{proposition}\label{prop:rank1com}
Let $\Bo$ be a rank-1 observable.
It is possible to measure an observable $\Ao$ without disturbing $\Bo$ if and only if $\Ao$ and $\Bo$ commute.
\end{proposition}

Suppose then that $\Ao$ is a rank-1 observable and it can be measured without disturbing another observable $\Bo$.
This does \emph{not} imply that $\Ao$ and $\Bo$ commute.
Indeed, the example given in Subsec. \ref{sec:not-symmetric} serves as a counterexample.

In spite of this fact, a measurement of a rank-1 observable does make all subsequent measurements useless.
In the following we make this statement precise and we characterize all instruments implementing rank-1 observables.

\begin{proposition}\label{prop:rank1}
Let $E$ be an effect. The following conditions are equivalent:
\begin{itemize}
\item[(i)] $E$ is rank-1.
\item[(ii)] Every completely positive linear mapping $\Phi$ on $\trh$ which satisfies $\Phi^\ast(\id)=E$ is of the form
\begin{equation}\label{eq:operation0}
\Phi(\varrho) = \tr{\varrho E}\xi
\end{equation}
for some state $\xi$.
\end{itemize}
\end{proposition}

\begin{proof}
(i)$\Rightarrow$(ii):
Let $\{K_j\}$ be the set of Kraus operators for $\Phi$, so that
\begin{equation}
\Phi(\varrho) = \sum_j K_j \varrho K^\ast_j \, , \quad \Phi^\ast(\id)=\sum_j K_j^\ast K_j=E \, .
\end{equation}
The last equation implies that for each $j$, $K_j^\ast K_j \leq E$.
Since $E$ is rank-1, there is a number $0< k_j \leq 1$ such that $K_j^\ast K_j = k_j E$. Clearly, $\sum_j k_j =1$.

Let $K_j = S_j \abs{K_j}$ be the polar decomposition of $K_j$.
Here $S_j$ is a partial isometry with $\ker S_j = \ker K_j$ and
\begin{equation*}
\abs{K_j} = \sqrt{K_j^\ast K_j} = \sqrt{k_j} \sqrt{E} =  \sqrt{ek_j} P \, .
\end{equation*}
For every state $\varrho$, we then get
\begin{equation*}
K_j \varrho K_j^\ast = e k_j  S_j P \varrho P S_j^\ast = e k_j \tr{\varrho P} S_jPS_j^\ast = \tr{\varrho E} k_j S_jPS_j^\ast
\end{equation*}
and hence
\begin{equation*}
\Phi(\varrho) = \tr{\varrho E} \sum_j k_j S_j P S_j^\ast \, .
\end{equation*}

The remaining thing is to show that the operator $S_j P S_j^\ast$ is a state for each $j$, implying that the convex sum $\sum_j k_j S_j P S_j^\ast =: \xi$ is a state also.
The operator $S_j P S_j^\ast$ is clearly positive. The operator $S_j^\ast S_j$ is the projection on the closure of $\ran \abs{K_j}$, thus $S_j^\ast S_j = P$.
Therefore,
\begin{equation*}
\tr{S_j P S_j^\ast} = \tr{S_j^\ast S_j P} = \tr{P}=1 \, .
\end{equation*}

(ii)$\Rightarrow$(i):
We assume that an effect $E$ is not rank-1 and show that there exists an operation satisfying $\Phi^\ast(\id)=E$ but not being of the form \eqref{eq:operation0}.
Let $E=\int_{\sigma(E)} \lambda d\Pi(\lambda)$ be the spectral decomposition of $E$.
We split the spectrum $\sigma(E)$ into two disjoint parts $\sigma_1$, $\sigma_2$ such that the operators $E_j=\int_{\sigma_j} \lambda d\Pi(\lambda)$, $j=1,2$ are nonzero.
Clearly, $E_1+E_2=E$ and both $E_1$ and $E_2$ have eigenvalue 0.
We fix two different states $\xi_1, \xi_2$ and define $\Phi_j(\cdot) = \tr{\cdot E_j}\xi_j$ for $j=1,2$.
Then the operation $\Phi=\Phi_1+\Phi_2$ satisfies $\Phi^\ast(\id)=E$.
Suppose that $\xi$ is a state satisfying \eqref{eq:operation0}.
Choose a unit vector $\psi_1$ such that $E_1\psi_1=0$ and $E_2\psi_1\neq 0$.
Then
\begin{equation*}
\Phi(\kb{\psi_1}{\psi_1}) = \ip{\psi_1}{E\psi_1}\xi = \ip{\psi_1}{E_2\psi_1}\xi
\end{equation*}
and, on the other hand,
\begin{equation*}
\Phi(\kb{\psi_1}{\psi_1}) = \Phi_1(\kb{\psi_1}{\psi_1}) + \Phi_2(\kb{\psi_1}{\psi_1}) = \ip{\psi_1}{E_2\psi_1}\xi_2 \, .
\end{equation*}
Hence, $\xi=\xi_2$.
But similarly we get $\xi = \xi_1$ if we repeat the calculation with a vector $\psi_2$ satisfying $E_2\psi_2=0$ and $E_1\psi_2\neq 0$.
This leads to the conclusion $\xi_1 = \xi_2$, which is a contradiction.
Therefore, $\xi$ does not exist.
\end{proof}

\begin{corollary}\label{prop:instrument0}
Let $\Ao$ be an observable. The following conditions are equivalent:
\begin{itemize}
\item[(i)] $\Ao$ is a rank-1 observable.
\item[(ii)] All instruments implementing $\Ao$ are of the form
\end{itemize}
\begin{equation}\label{eq:instrument0}
\Ii_x(\varrho) = \tr{\varrho \Ao_x}\xi_x \, ,
\end{equation}
where $\{\xi_x\}$ is a set of states.
\end{corollary}

Suppose we measure first a rank-1 observable $\Ao$ and then some other observable $\Bo$.
The instrument $\Ii$ describing the $\Ao$-measurement is of the form \eqref{eq:instrument0} and the joint probability distribution is thus given by
\begin{equation*}
\tr{\Ii_x(\varrho)\Bo_y} = \tr{\varrho\Ao_x} \tr{\xi_x \Bo_y} \, .
\end{equation*}
Here we see that the joint probability distribution can be calculated already after the first measurement since the numbers $\tr{\xi_x \Bo_y}$ do not depend on the initial state $\varrho$ at all.
Therefore, the $\Bo$-measurement is completely redundant.

\section{Measurements of the first kind}\label{sec:fkm}

A \emph{first kind measurement} is one which does not disturb itself.
More precisely, an instrument $\Ii$, implementing an observable $\Ao$, is a \emph{first kind instrument} if
\begin{equation}\label{eq:fk}
\Ic^\ast(\Ao_x)=\Ao_x\quad \textrm{for all}\ x\in\Omega_\Ao \, .
\end{equation}
We say that an observable $\Ao$ \emph{admits a first kind measurement} if there exists a first kind instrument which implements $\Ao$. In this section we study some conditions guaranteeing that an observable $\Ao$ admits or does not admit a first kind measurement. The overall picture is summarized in Fig. \ref{fig:2}.

A L\"uders instrument $\IiL$ implementing a commutative observable $\Ao$ satisfies the first kind condition \eqref{eq:fk}.
Therefore, all commutative observables admit first kind measurements.
As we have seen in Section \ref{sec:pairs}, commutativity is also a necessary condition if $\Ao$ is either a qubit observable (Prop. \ref{prop:qubit}) or a rank-1 observable (Prop. \ref{prop:rank1com}). 

An example of an observable not admitting a first kind measurement is an informationally complete observable.
Namely, any non-trivial measurement necessarily perturbs at least some state.
But an informationally complete observable gives different measurement outcome distributions for all states, hence a subsequent  measurement of the same observable detects any perturbation caused by the first measurement.
In a finite dimension $d\geq 3$ this observation is generalized by the condition that if $\dim\Ao\geq (d-1)^2+1$, then $\Ao$ does not admit a first kind measurement. Namely, $\Ao$ cannot be commutative since $(d-1)^2+1 > d$, and the conclusion therefore follows from Prop. \ref{prop:dim}.

\begin{figure}
\begin{center}
\includegraphics[width=12cm]{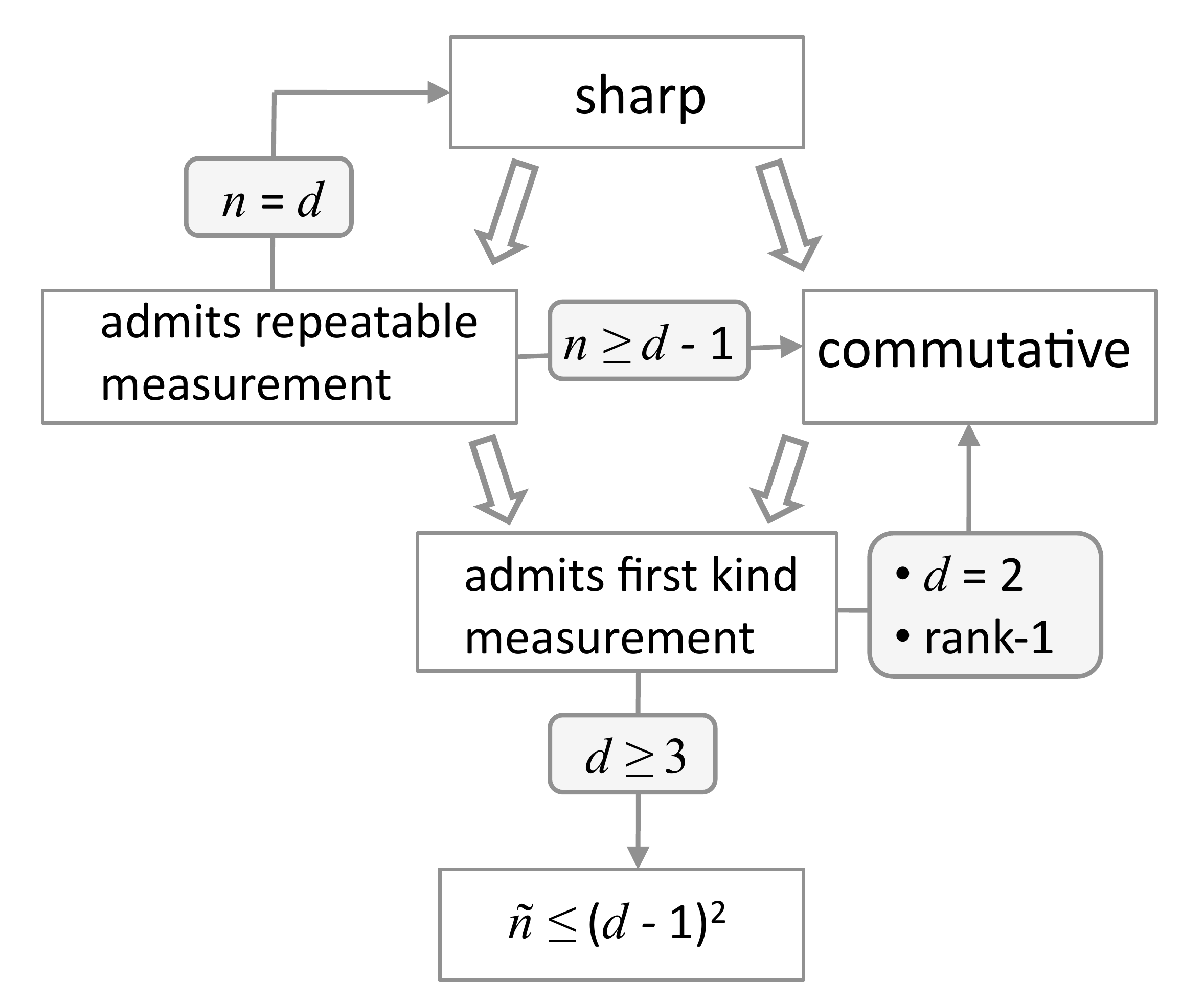}
\end{center}
\caption{Implications between the different concepts. Here $d=\dim\hi$, $n$ is the number of outcomes and $\widetilde{n}=\dim\Ao$. The gray boxes indicate the additional conditions under which the implications holds. Each of the two conditions in the bullet list is sufficient for the implication. \label{fig:2}}
\end{figure}

A more stringent condition than the first kind condition \eqref{eq:fk} is \emph{repeatability}; an instrument $\Ii$, implementing an observable $\Ao$, is repeatable if
\begin{equation}\label{eq:rep}
\Ii_x^\ast(\Ao_y)=0 \quad \textrm{whenever}\ x\neq y \, .
\end{equation}
This condition means that measuring repeatedly gives not only the same statistics but repeated measurement outcomes.

It is clear that a repeatable instrument is of the first kind.
The converse is, however, not true.
For instance, it is easy to see that a L\"uders instrument of a commutative observable is repeatable if and only if the associated observable is sharp.
But as we pointed out earlier, the L\"uders instrument of any commutative observable is first kind.

An observable $\Ao$ admits a repeatable measurement if and only if each effect $\Ao_x$ has an eigenvalue $1$ \cite{QTM96}.
Namely, it follows from \eqref{eq:rep} that whenever $\Ii_x(\varrho)\neq 0$, then $\Ii_x(\varrho)$ is an unnormalized eigenstate of $\Ao_x$ with eigenvalue $1$.
On the other hand, if $\Ao$ satisfies this eigenvalue condition, we can construct a repeatable instrument by first fixing for each outcome $x$ a unit vector $\psi_x$ satisfying $\Ao_x\psi_x=\psi_x$, and then defining
\begin{equation*}
\Ii_x(\varrho) = \tr{\varrho \Ao_x} \kb{\psi_x}{\psi_x} \, .
\end{equation*}
Since $\sum_x \Ao_x = \id$, we have $\Ao_y\psi_x=0$ whenever $y\neq x$, implying that $\Ii$ is repeatable.

We conclude that there are two sufficient conditions for an observable $\Ao$ to admit a first kind measurement:
\begin{itemize}
	\item $\Ao$ is commutative.
	\item Each effect $\Ao_x$ has eigenvalue $1$.
\end{itemize}
These two conditions overlap (e.g. sharp observables), but they also cover different situations.
To give an example of a non-commutative observable $\Ao$ having the property that each effect $\Ao_x$ has eigenvalue $1$, suppose that $\Omega_\Ao =\{1,2,3\}$ and $d=\dim\hi\geq 5$. We split $\hi=\hi_3\oplus \hi_{d-3}$ into a $3$ -dimensional and $d-3$ -dimensional subspaces $\hi_3$ and $\hi_{d-3}$.
In $\hi_3$ we fix three orthogonal one-dimensional projections $P_1,P_2,P_3$, while in $\hi_{d-3}$ we fix two non-commuting projections $R_1,R_2$.
Then $\Ao$, defined as
\begin{eqnarray*}
\Ao_1 = P_1 \oplus \half R_1 \, \quad \Ao_2 = P_2 \oplus \half R_2 \, , \\
\Ao_3 = P_3 \oplus ( \id -\half R_1 - \half R_2 ) \, ,
\end{eqnarray*}
has the required properties.

The previous example is actually minimal in the sense that a non-commutative observable must have at least 3 outcomes, and, as we show in Proposition \ref{prop:number-of-outcomes} below, a non-commutative observable can satisfy the eigenvalue condition only if $\mo{\Omega_\Ao} \leq d-2$.
In particular, for $d\leq 4$ only commutative observables can have repeatable measurements.

\begin{proposition}\label{prop:number-of-outcomes}
Suppose that $d=\dim\hi < \infty$ and let $\Ao$ be an obsevable admitting a repeatable measurement.
Then
\begin{itemize}
	\item[(a)] $\mo{\Omega_{\Ao}}\leq d$.
	\item[(b)]  If $\mo{\Omega_{\Ao}} = d$, then $\Ao$ is sharp.
	\item[(c)] If $\mo{\Omega_{\Ao}}\geq d-1$, then $\Ao$ is commutative.
\end{itemize}
\end{proposition}

\begin{proof}
For each $x$, we write $\Ao_x=P_x+R_x$, where $P_x$ is a one-dimensional projection and part of the
spectral projection of $\Ao_x$ associated with eigenvalue 1.
Due to normalization two projections $P_x,P_y$ must be orthogonal whenever $x\neq y$.
Since there can be only $d$ orthogonal projections in $\hi$, we conclude that (a) holds.

For every $x,y$, we have $P_x + R_y \leq \id$ (since $\sum_x \Ao_x = \id$).
Multiplying by $P_x$ on both sides gives $P_xR_yP_x \leq 0$, and thus $P_xR_yP_x = 0$.
As $R_y^2 \leq R_y$ (since $R_y\in\eh$), we get
\begin{equation*}
 0 \leq P_xR_y(P_xR_y)^\ast = P_xR_y^2P_x \leq P_xR_yP_x = 0 \, .	
\end{equation*}
Thus, $P_xR_y=0$ and $(\sum_x P_x)R_y =0$.
The sum $\sum_x P_x$ is a $\mo{\Omega_{\Ao}}$-dimensional projection.
If $\mo{\Omega_{\Ao}}=d$, then $R_y=0$. This proves (b).

If $\mo{\Omega_{\Ao}}=d-1$, then the previous calculation shows that all the effects $R_y$ are scalar multiples of the 1-dimensional projection $\id-\sum_x P_x$.
This proves (c).
\end{proof}

\section{Quantifying and deciding non-disturbance}\label{sec:sdp}

In this section we assume that $\dim\hi < \infty$ and that observables have only a finite number of outcomes.
We show how to decide whether an observable $\Ao$ can be measured without disturbing $\Bo$ and, if non-disturbance cannot be achieved, how to quantify the least amount of disturbance induced on $\Bo$ by measuring $\Ao$.

We  first propose a quantification of the non-disturbance relation in a form of a measure of disturbance $D_\Ao(\Bo)$ (Subsec. \ref{sec:sdp-1}).
This number satisfies $D_\Ao(\Bo)=0$ if and only if $\Ao$ can be measured without disturbing $\Bo$ and it has a simple physical interpretation.

We then demonstrate that calculating the number $D_\Ao(\Bo)$ is a semidefinite program (Subsec. \ref{sec:sdp-2}).
That is, the question whether an observable $\Ao$ can be measured without disturbing another observable $\Bo$ can be answered in an efficient and certifiable way. 
Note that this also implies that the question whether or not an observable admits a first kind measurement can be efficiently decided by setting $\Ao=\Bo$.

The problem of deciding whether or not observables are jointly measurable has been phrased in terms of a semidefinite program in \cite{WoPeFe09}.

\subsection{Quantification of the non-disturbance relation}\label{sec:sdp-1}

Let $\Ii$ be an instrument which implements an observable $\Ao$.
A second observable $\Bo$, measured after $\Ao$, is then possibly perturbed. 
Its measurement outcome probabilities, relative to the state prior to the $\Ao$-measurement, are given by the modified effects $\Ic^\ast(\Bo_y)$.
We want to quantify the difference between $\Bo$ and its perturbed version.

Let us first notice that there is always a number $\lambda\in[0,1]$ such that
 \begin{equation}\label{eq:difference}
 -\lambda \id \leq \Bo_y - \Ic^\ast(\Bo_y) \leq \lambda \id \qquad \forall y \, .
  \end{equation}
If $\Ao$ can be measured without disturbing $\Bo$, then we can choose $\Ii$ in a way that $\lambda=0$ in \eqref{eq:difference}.
Generally, the condition \eqref{eq:difference} is equivalent with the requirement that for every state $\varrho$,
 \begin{equation*}
\mo{\tr{\varrho\Bo_y} - \tr{\varrho\Ic^\ast(\Bo_y)}} \leq \lambda \qquad \forall y \, .
  \end{equation*}
This inequality is expressing that the measurement outcome probabilities differ at most by $\lambda$.

We conclude that the number
\begin{equation*}
\max_{y} \sup_{\varrho} \mo{\tr{\varrho\Bo_y} - \tr{\varrho\Ic^\ast(\Bo_y)}} = \max_y\no{\Bo_y - \Ic^\ast(\Bo_y)}
  \end{equation*}
 gives a natural quantification of the difference between $\Bo$ and its perturbed version $\Ic^\ast(\Bo)$.
In \cite{BuHe08} this kind of distance between two observables was used in the study of approximate joint measurability.

We now want to quantify the least amount of disturbance induced on $\Bo$ by an $\Ao$-measurement.
Hence, we denote by $D_\Ao(\Bo)$ the smallest number $\lambda$ attained by any implementation of $\Ao$, i.e.,
 \begin{equation}\label{eq:sdprimal1}
 D_\Ao(\Bo):=\inf\big\{\lambda\in[0,1] \ |\ \forall y:\;  -\lambda \id \leq \Bo_y - \Ic^\ast(\Bo_y) \leq \lambda \id\big\} \, ,
 \end{equation}
 where the infimum is taken over all instruments implementing $\Ao$.

 We will see in Subsec. \ref{sec:sdp-2} that the infimum in \eqref{eq:sdprimal1} is always attained.
 Therefore, $D_\Ao(\Bo)=0$ if and only if $\Ao$ can be implemented without disturbing $\Bo$. In general $D_\Ao(\Bo)$ is, by construction, the maximal possible disturbance of measured probabilities minimized over all instruments implementing $\Ao$.

 \begin{example}
 Let $\Ao$ and $\Bo$ be two sharp qubit observables.
We assume that they do not commute, hence an $\Ao$-measurement necessarily disturbs $\Bo$.
By Corollary \ref{prop:instrument0}, an instrument $\Ii$ implementing $\Ao$ is determined by two states $\xi_1$ and $\xi_2$, and we thus have
\begin{equation*}
\Ic^\ast(\Bo_y) = \tr{\xi_2\Bo_y}\id + \tr{(\xi_1-\xi_2)\Bo_y}\Ao_1 \, , \quad y=1,2 \, .
\end{equation*}
The two inequalities \eqref{eq:difference} for $y=1$ and $y=2$ are equivalent.
To find $D_\Ao(\Bo)$ we need to choose $\xi_1$ and $\xi_2$ in a way that the norm of the operator $\Bo_1 - \tr{\xi_2\Bo_1}\id - \tr{(\xi_1-\xi_2)\Bo_1}\Ao_1$ is as small as possible.
To calculate the norm, we write $\Ao_1 = \half ( \id + \vec{a}\cdot\vec{\sigma} )$ and $\Bo_1 = \half ( \id + \vec{b}\cdot\vec{\sigma} )$, where $\vec{a}, \vec{b}$ are unit vectors and $\vec{\sigma}=(\sigma_1,\sigma_2,\sigma_3)$ are the Pauli matrices. 
A straightforward calculation then shows that the smallest norm is $\half \sin\theta_{ab}$.
We thus conclude that $D_\Ao(\Bo)=\half \sin\theta_{ab}=\no{\Ao_1\Bo_1-\Bo_1\Ao_1}$.
The disturbance is therefore directly connected with the degree noncommutativity of $\Ao$ and $\Bo$.
\end{example}

\subsection{Non-disturbance as a semidefinite program}\label{sec:sdp-2}

In the following we show that deciding whether an observable can be measured without disturbing another observable is a semidefinite program.

To cast the quantification problem of Subsec. \ref{sec:sdp-1} into a semidefinite program, we notice that the task is to find the smallest number $\lambda$ under the conditions that there exists a collection $\{\Ii^\ast_x\}$ of linear maps satisfying
\begin{itemize}
	\item[(a)] $\Ii^\ast_x(\id) = \Ao_x$ for all $x$,
	\item[(b)] $\Ii_x^\ast$ is completely positive for all $x$,
	\item[(c)] $-\lambda \id \leq \Bo_y - \Ic^\ast(\Bo_y) \leq \lambda \id$ for all $y$.
\end{itemize}
The condition for completely positivity can be written as
\begin{equation*}
(\Ii_x\otimes{\rm id})(\omega)\geq 0 \, ,
\end{equation*}
where $\omega$ is the maximally entangled state $\omega=\frac{1}{d} \sum_{j,k} \kb{jj}{kk}$.

To proceed, we fix a selfadjoint operator basis $\{E_i\}$ satisfying $\tr{E_iE_j}=\delta_{ij}$ and $E_0=\frac{1}{\sqrt{d}}\id$.
Conditions (a)-(c) can now be written as follows:
\begin{itemize}
	\item[(a')] $\pm\frac{1}{\sqrt{d}} \Ao_x \leq \pm\sum_j \tr{E_j\Ii^\ast_x(E_0)} E_j  \quad \forall x$
	\item[(b')] $0\leq  \sum_{ij} \tr{E_j \Ii^\ast_x(E_i)} E_j \otimes E_i^T \quad \forall x$
	\item[(c')] $\pm\Bo_y \leq  \lambda \sqrt{d}E_0 \pm \sum_{ijx} \tr{E_j\Ii^\ast_x(E_i)} \tr{E_i\Bo_y} E_j  \quad \forall y $
\end{itemize}
To get (b') we have used the fact that the operator $\sum_i E_i\otimes E_i^T$ is
proportional to the maximally entangled state $\omega$.

We hence see that the task of minimizing $\lambda$ under conditions (a')-(c') can, after combining the constraints using a direct sum, be brought to the form
\begin{equation}\label{eq:sdp-primal}
\inf_{c\in\real^n} \left\{ \ip{v}{c} |\  F_0 \leq \sum_{\ell=1}^n c_\ell F_\ell \right\} \, ,
\end{equation}
where $v\in\real^n$ is a vector and $F_0,F_{\ell}$ are Hermitian matrices.
It is therefore a semidefinite program, and the dual problem is
\begin{equation}\label{eq:sdp-dual}
\sup_{C\geq 0} \left\{ \tr{F_0C} |\ \tr{ F_\ell C}=v_\ell \ \forall \ell \right\} \, .
\end{equation}

In the case under investigation, the dual takes the form
 \begin{equation*}
 D_\Ao^\ast(\Bo):= \sup \left( \sum_x \tr{H_x \Ao_x} - \sum_y \tr{K_y \Bo_y} \right)
 \end{equation*}
 where the supremum is taken over all selfadjoint operators $H_x$, $K_y$ satisfying
 \begin{itemize}
\item[(d)] $H_x\otimes\1\leq \sum_y K_y\otimes\Bo_y^T$ for all $x$
\item[(e)] $\sum_y \tr{\mo{K_y}}=1$.
\end{itemize}
By the general theory of semidefinite programs \cite{CO04}, $D_\Ao^\ast(\Bo)$ is automatically a lower bound on $D_\Ao(\Bo)$.
Actually, when written in the standard form \eqref{eq:sdp-dual} the dual program is seen to be strictly feasible.
It follows that $D_\Ao^\ast(\Bo)=D_\Ao(\Bo)$ and the extremum is attained for $D_\Ao(\Bo)$.
This also means that the proposed measure $D_{\Ao}(\Bo)$ can be efficiently computed numerically for any given pair $\Ao$ and $\Bo$, and that the obtained result can be certified by the solution of the dual.

\section*{Acknowledgements}

We acknowledge financial support by QUANTOP, the Danish research council (FNU) and the EU projects QUEVADIS and COQUIT.


\begin{thebibliography}{10}

\bibitem{ArGhGu02}
A.~Arias, A.~Gheondea, and S.~Gudder.
\newblock Fixed points of quantum operations.
\newblock {\em J. Math. Phys.}, 43:5872--5881, 2002.

\bibitem{Arveson72}
W.~Arveson.
\newblock Subalgebras of {$C^{\ast} $}-algebras. {II}.
\newblock {\em Acta Math.}, 128:271--308, 1972.

\bibitem{CO04}
S.~Boyd and L.~Vandenberghe.
\newblock {\em Convex optimization}.
\newblock Cambridge University Press, Cambridge, 2004.

\bibitem{BrJoKiWe00}
O.~Bratteli, P.E.T. Jorgensen, A.~Kishimoto, and R.F. Werner.
\newblock Pure states on {$O_d$}.
\newblock {\em J. Operator Theory}, 43:97--143, 2000.

\bibitem{BuHe08}
P.~Busch and T.~Heinosaari.
\newblock Approximate joint measurements of qubit observables.
\newblock {\em Quant. Inf. Comp.}, 8:0797--0818, 2008.

\bibitem{AlDo76}
S.T.~Ali and H.D.~Doebner.
\newblock On the equivalence of nonrelativistic quantum mechanics based upon sharp and fuzzy measurements.
\newblock {\em J. Math. Phys.}, 17:1105--1111, 1976.


\bibitem{QTM96}
P.~Busch, P.J. Lahti, and P.~Mittelstaedt.
\newblock {\em The Quantum Theory of Measurement}.
\newblock Springer-Verlag, Berlin, second revised edition, 1996.

\bibitem{BuSi98}
P.~Busch and J.~Singh.
\newblock L{\"u}ders theorem for unsharp quantum measurements.
\newblock {\em Phys. Lett. A}, 249:10--12, 1998.

\bibitem{QTOS76}
E.B. Davies.
\newblock {\em Quantum Theory of Open Systems}.
\newblock Academic Press, London, 1976.

\bibitem{HeReSt08}
T.~Heinosaari, D.~Reitzner, and P.~Stano.
\newblock Notes on joint measurability of quantum observables.
\newblock {\em Found. Phys.}, 38:1133--1147, 2008.

\bibitem{Lindblad99}
G.~Lindblad.
\newblock A general no-cloning theorem.
\newblock {\em Lett. Math. Phys.}, 47:189--196, 1999.

\bibitem{WeJu09}
L.~Weihua and W.~Junde.
\newblock On fixed points of {L}\"uders operation.
\newblock {\em J. Math. Phys.}, 50:103531, 2009.

\bibitem{Wolf10}
M.M. Wolf.
\newblock Quantum channels \& operations.
\newblock Lecture notes, available in $www.nbi.dk/\sim wolf/notes.pdf$, 2010.

\bibitem{WoPeFe09}
M.M. Wolf, D.~Perez-Garcia, and C.~Fernandez.
\newblock Measurements incompatible in quantum theory cannot be measured
  jointly in any other no-signaling theory.
\newblock {\em Phys. Rev. Lett.}, 103:230402, 2009.

\end{thebibliography}
\end{document}